%% file: arxiv/main.tex
\newif\ifarxiv
\arxivtrue 

\documentclass[11pt,letterpaper]{article}
\input{arxiv/macros}

\ifarxiv

\author{
 {Yossi Azar\thanks{Department of Computer Science, Tel-Aviv University,
    Tel-Aviv, Israel. Email: {\tt azar@tauex.tau.ac.il}. Supported in part by the Israel Science Foundation (grant No. 1728/26). }}
 \and
 {Debmalya Panigrahi\thanks{Department of Computer Science, Duke University, Durham, NC. Email: {\tt debmalya@cs.duke.edu}. Supported in part by NSF grants CCF-2329230 and CCF-1955703.}}
 \and
 {Or Vardi\thanks{Department of Computer Science, Tel-Aviv University,
    Tel-Aviv, Israel. Email: {\tt orvardi@mail.tau.ac.il}.}} 
}

\else

\author{}

\fi

\date{}

\title{Online Metric TSP: Beyond the $\sqrt{n}$ Barrier}

\begin{document}

\maketitle

\begin{abstract}
We study an online variant of the Traveling Salesperson Problem (TSP) in which $n$ points arrive sequentially and must be inserted into an evolving tour. In the classical setting where arbitrary insertions are allowed, an $O(\log n)$-competitive algorithm has been known since the 1970s (Rosenkrantz, Stearns and Lewis 1977, Imase and Waxman 1991). Recently, Abrahamsen, Bercea, Beretta, Klausen, and Kozma [ESA 2024] introduced online metric TSP, a stricter model in which each arriving point must be assigned to a distinct cell of an array of size $m \ge n$, with the final tour order induced by the non-empty cells; the parameter $m$ captures the space usage of the algorithm.

When $m = 2^{n}$, this model recovers arbitrary insertions and therefore admits an $O(\log n)$-competitive algorithm. In contrast, when $m = n$, i.e., when each point’s position is fixed on arrival, Bertram~\cite{bertram2025onlinemetrictsp} recently showed that the competitive ratio is $\Theta(\sqrt{n})$. We investigate the tradeoff between space usage and competitiveness between these extremes. We note that this tradeoff was previously explored by the authors in \cite{AzarPV26} for the online sorting problem, which is the special case of online metric TSP on a line metric.

Our main result is a deterministic online metric TSP algorithm using $m = (1+\varepsilon) n$ space that achieves a competitive ratio of $O(\log^{3} n / \varepsilon)$, for any $\varepsilon  \le 1$. In particular, increasing the space from $n$ to $2n$ improves the competitive ratio from $\Theta(\sqrt{n})$ to $O(\log^{3} n)$. We complement this with a lower bound showing that for $m = n^{1+\varepsilon}$, any deterministic algorithm has a competitive ratio $\Omega(1/\varepsilon)$, for all $\varepsilon \ge \Omega(\log \log n / \log n)$. Consequently, even with $m = O(n \cdot \mathrm{polylog}(n))$, deterministic algorithms cannot achieve a constant competitive ratio.
\end{abstract}

\pagenumbering{gobble}
\clearpage
\pagenumbering{arabic}
\setcounter{page}{1}

\newcommand{\mbc}{\textsc{MBC}\xspace}

\section{Introduction}\label{sec:intro}\input{arxiv/intro}

\section{Online Metric TSP using Metric Ball-Cover (\mbc) Trees
}\label{sec:elementary}\input{arxiv/elementary}

\section{Lower Bound for Online Sorting and Online Metric TSP}\label{sec:lower}\input{arxiv/small-error-lower-bound}

 \section{Closing Remarks}\label{sec:closing}
 \input{arxiv/closing-remarks}
\
{\small
\bibliography{arxiv/bib}
}

\appendix

\section{The general case for unknown $\opt$}\label{sec:doubling}\input{arxiv/general-metric-doubling}

\section{A Counter-Example to the Analysis from \cite{AzarPV26} for a General Metric Space}\label{sec:counterexample}\input{arxiv/counterexample}

\end{document}

%% file: arxiv/macros.tex
\usepackage{typearea}
\usepackage{setspace}

\usepackage{fullpage}

\usepackage{comment} 
\usepackage{bbm}
\usepackage{framed} 
\usepackage{url} 
\usepackage{complexity}
\usepackage{booktabs}
\usepackage{amsmath,amssymb}
\usepackage{float}

\usepackage{amsthm}
\usepackage{thmtools} 
\usepackage{thm-restate}
\usepackage{nicefrac}
\usepackage{calc}
\usepackage{enumerate}
\usepackage{enumitem}
\usepackage[usenames,dvipsnames]{xcolor}

\usepackage[ruled,vlined,linesnumbered]{algorithm2e}
\usepackage{forloop}

\usepackage{graphicx}
\usepackage[font=footnotesize,labelfont=bf]{subcaption}
\usepackage[font=footnotesize,labelfont=bf]{caption}
\usepackage[nobreak=true]{mdframed}
\usepackage{appendix}
\usepackage[noend]{algpseudocode}
\usepackage[colorinlistoftodos]{todonotes}
\usepackage{xr}
\usepackage{array}
\usepackage{xspace}
\definecolor{ForestGreen}{rgb}{0.1333,0.5451,0.1333}
\definecolor{DarkRed}{rgb}{0.8,0,0}
\definecolor{Red}{rgb}{1,0,0}
\usepackage[linktocpage=true,
pagebackref=true,colorlinks,
linkcolor=DarkRed,citecolor=ForestGreen,
bookmarks,bookmarksopen,bookmarksnumbered]{hyperref}

\usepackage[capitalise]{cleveref}
\crefrangelabelformat{enumi}{#3#1#4--#5#2#6}
\usepackage{parskip}
\usepackage{chngcntr}
\usepackage{mathtools,stackengine}
\usepackage{multirow}

\stackMath
\newcommand{\stackGeq}[1]{%
	\setbox0=\hbox{${}\mathrel{\stackon[-1pt]{\geq}{\scriptstyle\text{#1\strut}}}{}$}
	\xdef\tmpwd{\dimexpr\the\wd0\relax}
	\kern.5\tmpwd\mathclap{\box0}&\kern.5\tmpwd
}

\usepackage{nameref}

\usepackage{tikz}
\usetikzlibrary{patterns}

\allowdisplaybreaks

\let\polylog\relax
\DeclareMathOperator*{\polylog}{polylog}

\newcommand{\eat}[1]{}

\newcommand{\ins}{{\tt insert}}

\newcommand\eps{\varepsilon}

\DeclarePairedDelimiterX{\expectarg}[1]{[}{]}{%
	\ifnum\currentgrouptype=16 \else\begingroup\fi
	\activatebar#1
	\ifnum\currentgrouptype=16 \else\endgroup\fi
}

\DeclarePairedDelimiterX{\nicesetarg}[1]{\{}{\}}{%
	\ifnum\currentgrouptype=16 \else\begingroup\fi
	\activatebar#1
	\ifnum\currentgrouptype=16 \else\endgroup\fi
}

\newcommand{\innermid}{\nonscript\;\delimsize\vert\nonscript\;}
\newcommand{\activatebar}{%
	\begingroup\lccode`\~=`\|
	\lowercase{\endgroup\let~}\innermid 
	\mathcode`|=\string"8000
}

\newcommand\opt{\mathsf{opt}\xspace}

\counterwithin{equation}{section}

\usepackage{eqparbox}

\theoremstyle{plain}

\newtheorem{theorem}{Theorem}[section]

\newtheorem{lemma}[theorem]{Lemma}

\newtheorem{claim}[theorem]{Claim}

\newtheorem{remark}[theorem]{Remark}

\newlength{\continueindent}
\usepackage{etoolbox}
\makeatletter
\newcommand*{\ALG@customparshape}{\parshape 2 \leftmargin \linewidth \dimexpr\ALG@tlm+\continueindent\relax \dimexpr\linewidth+\leftmargin-\ALG@tlm-\continueindent\relax}
\apptocmd{\ALG@beginblock}{\ALG@customparshape}{}{\errmessage{failed to patch}}
\makeatother

\makeatletter
\def\thm@space@setup{%
	\thm@preskip=\parskip \thm@postskip=0pt
}
\makeatother

\usepackage{etoolbox}
\usepackage{tikz}
\usetikzlibrary{tikzmark}
\usetikzlibrary{calc}

\errorcontextlines\maxdimen

\newcommand{\ALGtikzmarkcolor}{black}
\newcommand{\ALGtikzmarkextraindent}{4pt}
\newcommand{\ALGtikzmarkverticaloffsetstart}{-.5ex}
\newcommand{\ALGtikzmarkverticaloffsetend}{-.5ex}
\makeatletter
\newcounter{ALG@tikzmark@tempcnta}

\newcommand\ALG@tikzmark@start{%
	\global\let\ALG@tikzmark@last\ALG@tikzmark@starttext%
	\expandafter\edef\csname ALG@tikzmark@\theALG@nested\endcsname{\theALG@tikzmark@tempcnta}%
	\tikzmark{ALG@tikzmark@start@\csname ALG@tikzmark@\theALG@nested\endcsname}%
	\addtocounter{ALG@tikzmark@tempcnta}{1}%
}

\def\ALG@tikzmark@starttext{start}
\newcommand\ALG@tikzmark@end{%
	\ifx\ALG@tikzmark@last\ALG@tikzmark@starttext
	\else
	\tikzmark{ALG@tikzmark@end@\csname ALG@tikzmark@\theALG@nested\endcsname}%
	\tikz[overlay,remember picture] \draw[\ALGtikzmarkcolor] let \p{S}=($(pic cs:ALG@tikzmark@start@\csname ALG@tikzmark@\theALG@nested\endcsname)+(\ALGtikzmarkextraindent,\ALGtikzmarkverticaloffsetstart)$), \p{E}=($(pic cs:ALG@tikzmark@end@\csname ALG@tikzmark@\theALG@nested\endcsname)+(\ALGtikzmarkextraindent,\ALGtikzmarkverticaloffsetend)$) in (\x{S},\y{S})--(\x{S},\y{E});%
	\fi
	\gdef\ALG@tikzmark@last{end}%
}

\apptocmd{\ALG@beginblock}{\ALG@tikzmark@start}{}{\errmessage{failed to patch}}
\pretocmd{\ALG@endblock}{\ALG@tikzmark@end}{}{\errmessage{failed to patch}}
\makeatother

\algblock[with]{With}{EndWith}
\algblockdefx[With]{With}{EndWith}%
[1]{\textbf{with} #1 \textbf{do}}%
{}

\makeatletter
\ifthenelse{\equal{\ALG@noend}{t}}%
{\algtext*{EndWith}}
{}%
\makeatother

%% file: arxiv/intro.tex
The (metric) traveling salesman problem (TSP) is a classic problem in combinatorial optimization where the goal is to find the minimum-length tour covering a set of $n$ points in a metric space. In a classic result in approximation algorithms, Christofides~\cite{Christofides76} (and independently, Serdyukov~\cite{Serdyukov76}) gave a $3/2$-approximation to this problem, which remained the state-of-the-art for almost 50 years before being eventually improved to $3/2-\eps$ (for a small $\eps > 0$) in a remarkable recent result of Karlin, Klein, and Oveis Gharan~\cite{KarlinKO24}. On the hardness side, Papadimitrou and Yannakakis~\cite{PapadimitriouY93} showed APX-hardness for this problem as a consequence of the PCP theorem~\cite{AroraLMSS98}; the current record on the lower bound stands at $123/122$~\cite{KarpinskiLS15}. Resolving this gap remains a central open question in approximation algorithms. 

What if the set of points that we need to connect is not known in advance but arrives online? In each online step, the goal is to insert the newly arriving point in the existing order in a way that minimizes the length of the tour following that order. In the classical version of this problem, where the arriving point can be inserted anywhere in the existing order, $O(\log n)$-competitive algorithms are known since the 1970s~\cite{RosenkrantzSL77,ImaseW91}. Recently, Abrahamsen, Bercea, Beretta, Klausen, and Kozma~\cite{AbrahamsenB0K024} introduced a stricter model called online metric TSP: here, the online algorithm has to map each arriving point to a distinct cell of an array of size $m \ge n$, and the eventual order is given by the sequence of non-empty cells in the array. (We refer to the parameter $m$ as the space usage of the algorithm.) If $m = 2^n$, this setting recovers arbitrary insertions thereby admitting a competitive ratio of $O(\log n)$. At the other extreme, when $m=n$, the online algorithm commits to an irrevocable position of the new point in the eventual order immediately on arrival. For this setting, Bertram~\cite{bertram2025onlinemetrictsp} recently gave a deterministic algorithm with a competitive ratio of $O(\sqrt{n})$, matching previously known lower bounds even allowing randomization~\cite{AamandA0K23,AbrahamsenB0K024}. The exponential gap in the competitive ratio between the two extremes raises a natural question about what lies in between, namely: {\em how does the competitive ratio of online metric TSP depend on the space usage $m$?}

We note that this tradeoff between competitive ratio and space usage was previously explored by the authors for the {\em online sorting} problem~\cite{AzarPV26} introduced by Aamand, Abrahamsen, Beretta, and Kleist~\cite{AamandA0K23}, which is the special case of online metric TSP on a line metric. We will describe later that the solution proposed by \cite{AzarPV26} for online sorting does not generalize naturally to online metric TSP; indeed, the algorithm given in this paper is a different, and arguably more natural, solution for the online sorting problem as well. Moreover, the lower bound in this paper is applicable to the online sorting problem as well. Therefore, in addition to studying online metric TSP with $m > n$, this paper makes contributions to the online sorting problem for the case of $m > n$ as well.

\subsection{Our Results}

In this paper, we answer this question by showing a poly-logarithmic competitive ratio for online metric TSP even when $m$ only mildly exceeds $n$. Note that this is an exponential improvement over the case $m=n$. Specifically, we show:

\begin{theorem}\label{thm:upper}
    There is a deterministic online metric TSP algorithm that uses $m = (1+\eps) n$ space and achieves a competitive ratio of  $O(\log^3 n / \eps)$, for any $\eps \in [O(\log n / n), 1]$.
\end{theorem}
\noindent{\bf Remark:} \cref{thm:upper} was obtained concurrently and independently by Aamand et al. \cite{aamand2026onlinegeometricpackingonline}.

This theorem implies, for instance, that with $n/\polylog(n) = o(n)$ extra space, the competitive ratio of online metric TSP improves from $\Theta(\sqrt{n})$ to $\polylog(n)$. 
Prior to our work, the only cases for which a sub-polynomial competitive ratio was known had a space usage of $m = 2^n$. It is interesting to ask whether the extra space can be reduced even further to $n^{1-\delta}$ for any constant $\delta > 0$, or whether there is a lower bound on the extra space required to obtain a $\polylog(n)$ competitive ratio. We leave these are interesting open questions.

Can the competitive ratio be improved further to $O(1)$? To the best of our knowledge, such a result is not known even for unlimited space, but there is no lower bound either, even for $m = O(n \cdot log n)$. We complement our upper bound by showing that it is impossible for a deterministic algorithm to achieve a competitive ratio of $O(1)$, unless $m$ polynomially exceeds $n$. We show:

\begin{theorem}\label{thm:lower}
    For $m = n^{1+\eps}$, the competitive ratio of any deterministic online metric TSP algorithm is $\Omega(1/\eps)$, for any $\eps > \Omega(\log \log n/\log n)$.
\end{theorem}

This implies that using even $m = O(n\cdot \polylog(n))$, the best competitive ratio one can potentially achieve deterministically is $O(\log n / \log\log n)$. 

\noindent
{\bf Remark:} We note that the lower bound construction that we will give later to prove \Cref{thm:lower} is actually for the online sorting problem, which is a special case of onine metric TSP. In \cite{AzarPV26}, the authors showed that using $m = O(n \log^2 n)$, one can achieve a competitive ratio of $O(1)$ for the online sorting problem. However, this result does not contradict the lower bound in this paper since the previous result assumed that the algorithm knows the value of $\opt$ from the outset, while the lower bound is for the case where the value of $\opt$ is unknown to the algorithm. It remains open whether an $O(1)$-competitive algorithm exists for the online metric TSP problem for the case of known $\opt$.

\subsection{Related Work} 

Online versions of TSP have been studied since the 1970s when Rosenkrantz, Stearns, and Lewis~\cite{RosenkrantzSL77} introduced the class of constructive insertion heuristics to iteratively build a TSP tour by inserting new vertices while preserving the relative order of previous vertices. They gave a $2$-approximate algorithm when the insertion order is chosen by the algorithm, and an $O(\log n)$-competitive one when it is chosen by an adversary. The latter result also follows from later work by Imase and Waxman on the online Steiner tree problem~\cite{ImaseW91}, where the online algorithm additionally commits to the edges being used to connect the new vertex to the previous vertices. This latter commitment also creates a difference between the two problems: while $O(\log n)$-competitiveness is tight for online Steiner tree, it is not known to be tight for creating an online TSP tour. For the latter bound, the only lower bounds known are for specific insertion strategies, such as for the greedy strategy due to Azar~\cite{Azar94}. (See also \cite{MegowSVW16} for a recourse version of online MST and TSP where the algorithm commits to edges but these can be partially revoked using a limited budget.)

Later, a different online version of TSP was proposed, where a server has to serve a sequence of requests arriving online on a metric space and the goal is to minimize the makespan, i.e., the total time to serve all requests~\cite{AusielloFLST01}. This problem differs from prior work (and from online metric TSP) in that there is a real notion of time with release dates for the server requests. We note that several results have been obtained in this line of work, both for general metric spaces~\cite{FeuersteinS01,BlomKPS01,KrumkeLLMPPS02,KrumkePPS03,Lipmann03,JailletW08} and also specifically for the line metric~\cite{BjeldeHDHLMSSS21}.

Online metric TSP was introduced by Abrahamsen, Bercea, Beretta, Klausen, and Kozma~\cite{AbrahamsenB0K024} as a generalization of the online sorting problem~\cite{AamandA0K23} which restricts the points to a line. It is known that online sorting has a tight competitive ratio of $\Theta(\sqrt{n})$ when $m=n$~\cite{AamandA0K23,AbrahamsenB0K024}, and that this bound continues to hold for online metric TSP with $m=n$~\cite{bertram2025onlinemetrictsp}. For $m > n$, an $O(1)$-competitive algorithm is known for $m = O(n \cdot \polylog(n))$ when $\opt$ is known, and an $O(\log n)$-competitive algorithm is known for $m = (1+\eps) n$ (for $\eps > 0$) when $\opt$ is unknown~\cite{AzarPV26}. Note that the lower bound result in the current paper also applies to online sorting; hence, it shows that this distinction between known and unknown $\opt$ is necessary. Further results on the online sorting problem appear in \cite{NirjhorW25,improved-stochastic-hu,FotakisKPT25}.

\subsection{Our Techniques} 

We now give an overview of our main techniques. We first describe our upper bound (\Cref{thm:upper}) and then outline the lower bound construction (\Cref{thm:lower}). 

\smallskip\noindent
{\bf Upper Bound.}
Our starting point is the previous result for the online sorting problem with $m > n$~\cite{AzarPV26}. This paper gave an $O(\log^3 n/\eps)$-competitive algorithm using $m=(1+\eps)n$ space, based on a novel data structure called an \emph{elementary tree}. The core idea is to maintain a collection of binary trees whose leaves, ordered from left to right, correspond to array cells, and which guide the insertion of arriving elements.

We briefly recall how elementary trees are used to insert elements. The algorithm dynamically and sequentially labels internal tree nodes with dyadic intervals, whose lengths correspond to the heights of the nodes. For example, if the element range is $[0,1]$, the root is labeled $[0,1]$, but its two children may both be labeled $[0,1/2]$, or $[1/2,1]$, or one of each. This continues down the tree with the interval length halving at each level. The label of a node $v$ specifies the range of element values that may occupy the array cells corresponding to the leaves of the subtree rooted at $v$. Crucially, except at the root, these labels are assigned dynamically and depend on the insertion sequence.

When a new element arrives, the algorithm performs a depth-first search to find a node whose interval label contains the element and which has an unlabeled child. That child is then labeled with the unique dyadic interval at that level containing the element, and the procedure recurses down the tree until the element is placed at a leaf. This approach fundamentally exploits the shared linear order of the line metric and the array, a property that does not extend to general metric spaces.

Extending this framework to general metrics faces several obstacles. First, unlike dyadic intervals, recursive covers (by balls of decreasing radii) of general metric spaces typically involve overlapping regions, which violates laminarity. (Note that laminarity requires that any two of the covering sets should either be disjoint or one should be contained in the other.) As a result, a point may not belong to a unique ball at a given level, introducing ambiguity that the insertion algorithm must resolve. Second, while dyadic intervals split into exactly two subintervals at each level, a ball in a general metric space may require many smaller balls to cover it at the next scale, with the number governed by the doubling dimension. Since elementary trees are binary, this creates a mismatch between the recursive structure of the metric space and the tree arity.\footnote{In fact, the doubling dimension can be arbitrarily large for a general metric space.} Increasing the arity is also  generally infeasible, as it would lead to excessive space usage. Finally, there is a more fundamental disconnect: while dyadic intervals inherit a natural left-to-right order from the line metric compatible with the array, there is no canonical linear ordering of balls of the same radius in a general metric space.

\smallskip\noindent{\bf Metric Ball-Cover (\mbc) Tree.}
To overcome these difficulties, we introduce a new data structure in this paper called a {\em metric ball-cover tree} (\mbc tree). Our first step is to construct a {\em dynamic} recursive cover of the metric space using balls of geometrically decreasing radii. Unlike dyadic intervals on the line that are statically defined, the balls in our cover are defined based on the arrival sequence. When a new point arrives, the algorithm identifies the highest level at which it is not already covered by an existing ball and creates balls of progressively smaller radii from that level onward, all centered at the new point.

Given this (dynamically evolving) ball-cover of the metric space, an \mbc tree is a binary tree where the nodes are {\em dynamically labeled} by points in the metric space. Labeling a node $v$ at height $h$ with a point $p$ in the metric space commits the array cells corresponding to the leaves under $v$ in the \mbc tree to points in a ball of radius $r$ centered at $p$, where $r$ is a fixed function of $h$. 

This brings us to our main challenge: we need to define the rules of labeling an \mbc tree culminating in the insertion of a new point in the array. A natural attempt would be to reuse the depth-first insertion procedure of~\cite{AzarPV26}, but as we show in \Cref{sec:counterexample}, this approach fails to preserve the desired guarantees in general metric spaces.
Instead, we introduce a new \emph{bottom-up} labeling strategy. To motivate it, consider the risk associated with labeling a node $v$ by a ball of radius $r$ centered at a newly arrived point $p$. Since labels are irrevocable, this commits the entire subtree under $v$ to points within that ball, even though only a single point has been observed so far. If no other points ever fall into this ball, the remaining array cells in that subtree are wasted.

This phenomenon is significantly more severe in general metric spaces than on the line. Suppose we call a subtree that receives only a single insertion a \emph{wasted} subtree. Note that each ball/interval can only correspond to a single wasted subtree, since any subsequent insertion in the ball/interval would not label a new node at the same level. For dyadic intervals on the line, this is a strong property: it immediately implies that the total wasted space is cumulatively no more than the size of the subarray under a single elementary tree, since each dyadic interval splits into exactly two subintervals at the next level. Indeed, by constructing elementary trees with $\eps n$ leaves, the algorithm of \cite{AzarPV26} ensures that the total wasted space is $\eps n$.

In contrast, in a general metric space, a ball may split into many smaller balls at the next level. Since the tree remains binary, this can cause the total wasted space to grow well beyond the size of a single tree. Increasing the tree arity to match the doubling dimension does not resolve this issue, as it merely coalesces the wasted space to a larger subarray that now corresponds to the leaves of an elementary tree of fixed height.

To overcome this bottleneck, we adopt a frugal allocation strategy in \mbc trees: we label the \emph{deepest} unlabeled node that can legitimately be assigned to the new point. This bottom-up approach minimizes the amount of space committed based on limited evidence. However, abandoning the depth-first insertion procedure in \cite{AzarPV26} necessitates a fundamentally new analysis of both space usage and competitive ratio of the algorithm. We describe the \mbc tree data structure and the online metric TSP algorithm in \Cref{sec:elementary}. For simplicity, we first assume that the optimal value $\opt$ is known to the algorithm. Under this assumption, we obtain an $O(\log^2 n/\eps)$-competitive algorithm using $m=(1+\eps)n$ space. Removing the assumption that $\opt$ is known incurs an additional $O(\log n)$ factor in the competitive ratio, following techniques similar to~\cite{AzarPV26}. This extension is presented in \Cref{sec:doubling}, completing the proof of \Cref{thm:upper}.

\smallskip\noindent{\bf Lower Bound.}
Our lower bound construction (\Cref{thm:lower}) applies to online sorting~\cite{AamandA0K23}, the special case of online metric TSP on the line. A key feature of the construction is that the optimal value $\opt$ is \emph{not} known to the online algorithm. While this assumption is standard in online metric TSP, some prior work on online sorting assumes that all elements lie in $[0,1]$ and that both endpoints appear in the instance, effectively fixing $\opt=1$. Our lower bound does not apply in this setting; indeed, when $\opt$ is known, $O(1)$-competitive algorithms with $m=O(n\polylog n)$ space are known for online sorting~\cite{AzarPV26}.

Recall that in online sorting, elements arrive online and must be placed into an array of size $m$ so as to minimize the sum of absolute differences between consecutive non-empty cells. The lower bound proceeds in multiple phases. Suppose in the first phase, the adversary presents batches of elements that are evenly spaced within the interval $[1,2]$. The algorithm’s behavior can be summarized by the structure of the \emph{holes}, i.e., the empty regions between occupied array cells. If the algorithm keeps these holes small, it is being frugal with space, but the adversary responds by increasing the batch sizes while keeping the value range fixed. Maintaining small holes under such dense arrivals forces the algorithm to incur a large competitive ratio, as it repeatedly pays cost at the boundaries between batches.

Eventually, to control its competitive ratio, the algorithm must create sufficiently large holes. At this point, the adversary switches strategy and begins presenting elements in a much larger range $[k,2k]$, for large $k$. This effectively turns the earlier elements into zeros relative to the new scale. Since the algorithm lacks sufficient unused space, it must reuse the previously created large holes. Inserting large values into these holes immediately incurs cost $\Theta(k)$ at the boundaries, which matches the order of $\opt$ at this stage. If this happens in too many holes, the algorithm again suffers a large competitive ratio.

Our construction, given in full in \Cref{sec:lower}, iterates this process over multiple rounds and refines the idea presented above to obtain the lower bound stated in \Cref{thm:lower}.

%% file: arxiv/elementary.tex
We introduce a new data structure called a Metric Ball-Cover (\mbc) tree. In this section, we describe \mbc trees and show how they are used to solve online metric TSP. 

An \mbc tree of height $H$ is an array data structure of size $2^H$ that has an associated complete binary tree of height $H$. We emphasize that the tree is \emph{virtual} in the sense that it is used to define how points are inserted into the array, but the actual data structure is just an array. The leaves of the tree, at height $0$, are each associated with a unique array cell, enumerated from left to right. The non-leaf nodes of the tree do not correspond to array cells, and are only used in the algorithm to decide the location where an arriving point in inserted in the array. Thus, all insertions of points are at the leaves of the tree. The root of the tree is at height $H$, its children are at height $H-1$, their children at height $H-2$, and so on.

\smallskip\noindent{\bf Node Labeling.} 
Each node $v$ at height $h$ is associated with a ball $B(c_v, r_h)$. 
\begin{itemize}
    \item[-] \textbf{Radius ($r_h$):} The radius is deterministic and depends only on the height. We set
        $r_h := R \cdot 2^{h-H}$,
    where $R$ is a scaling factor denoting the radius of the ball at the root, which will be defined when the tree is created.
    \item[-] \textbf{Center ($c_v$):} The center is determined dynamically. We say a node in the tree is \emph{marked} if a center $c_v$ has been assigned to it. Otherwise, the node is \emph{unmarked}. Initially, all nodes are unmarked. A marked node $v$ is \emph{partial} if it has one marked child and one unmarked child. For a marked node $v$, we say that a point $x_t$ is \emph{admissible} at $v$ if $d(x_t, c_v) \le r_h$ where $h$ is the height of $v$.
\end{itemize}

An example of an \mbc tree is shown in \Cref{fig:mbc-tree}.

\smallskip\noindent{\bf The Data Structure.}
The overall data structure comprises a sequence of identical \mbc trees $A_0, A_1, \ldots$, each of fixed height 
    $H := \lfloor\log (\eps n) - \log \log n\rfloor$,\footnote{All logarithms are with base $2$ unless otherwise mentioned.}   
and with the roots labeled with balls of radius $R$. The number of \mbc trees is not fixed in advance; the sequence grows dynamically as needed.

\begin{remark}
    \label{rem:epsilon-cannot-be-too-small}
    We assume $\eps \ge \log n / n$, ensuring that the height $H$ of the \mbc trees is non-negative.
\end{remark}

\subsection{The Online Metric TSP Algorithm}
We assume knowledge of the value of the optimal solution $\opt$ in describing our algorithm in this section. In \Cref{sec:doubling}, we will relax this assumption at the cost of an additional $O(\log n)$ factor in the competitive ratio. We use $\opt$ to set the radius of the balls at the roots of the \mbc trees in our data structure, i.e., $R := \opt$ in our data structure.

We now describe the online algorithm that inserts a point $x_t$ into the array using the sequence of \mbc trees.
The algorithm searches for a \textbf{partial, admissible} node of the \textbf{lowest possible height} across all existing \mbc trees.
\begin{itemize}
    \item If such a node $v$ is found, we call $\ins(v, x_t)$, which is described below.
    \item If no such node exists, we create a new \mbc tree $A_{\text{new}}$, add it to the sequence, and call $\ins(\text{root}(A_{\text{new}}), x_t)$.
\end{itemize}

\smallskip\noindent{\bf Local Insertion ($\ins(v, x)$).}
The procedure starts by identifying a node $v'$ as follows:
\begin{itemize}
    \item If $v$ is a \emph{partial node}, $v'$ is the \emph{right child} of $v$.
    \item Otherwise, $v$ is the root of a new tree; in this case, $v' := v$.
\end{itemize}
Starting from $v'$, the algorithm traverses down to a leaf by repeatedly choosing the \emph{left} child. For every node $w$ on this path (including $v'$), we mark $w$ with center $c_w = x$ and radius $r_h$ where $h$ is the height of $w$. Finally, $x$ is stored at the leaf $u$ at the end of this path.

In \Cref{fig:insertion to mbc}, we show the insertion of a new point in an \mbc tree.

\begin{figure}[tbh]
\centering
\includegraphics[width=0.8\linewidth]{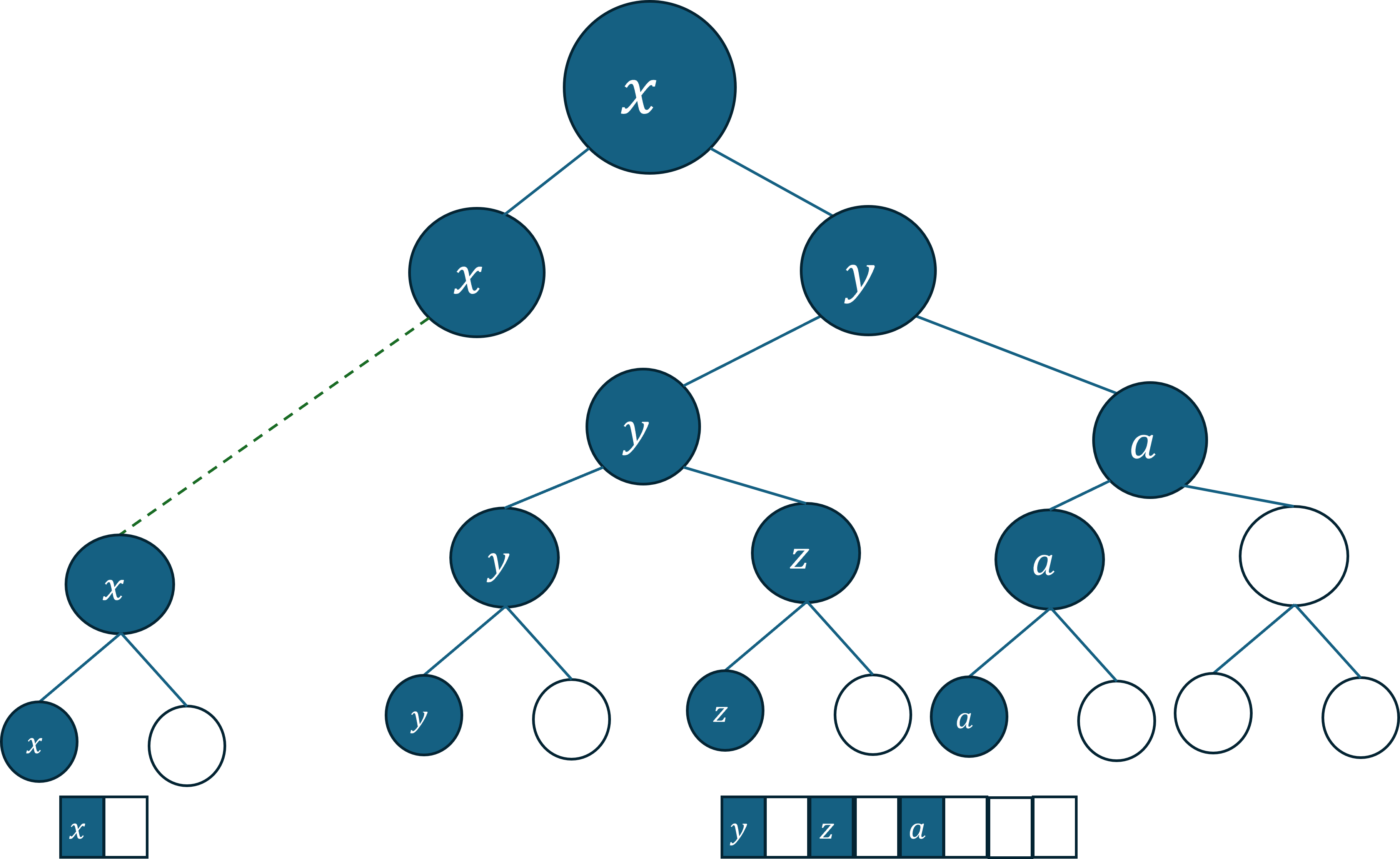}
    \caption{Illustration of an \mbc tree. The colored cells in the array represent occupied cells. Similarly, the colored nodes in the tree represent marked nodes.}
\label{fig:mbc-tree}    
\end{figure}

\begin{figure}[tbh]
    \centering
    \subfloat[\centering Before the insertion of $x_t$]{{\includegraphics[width=.49\linewidth]{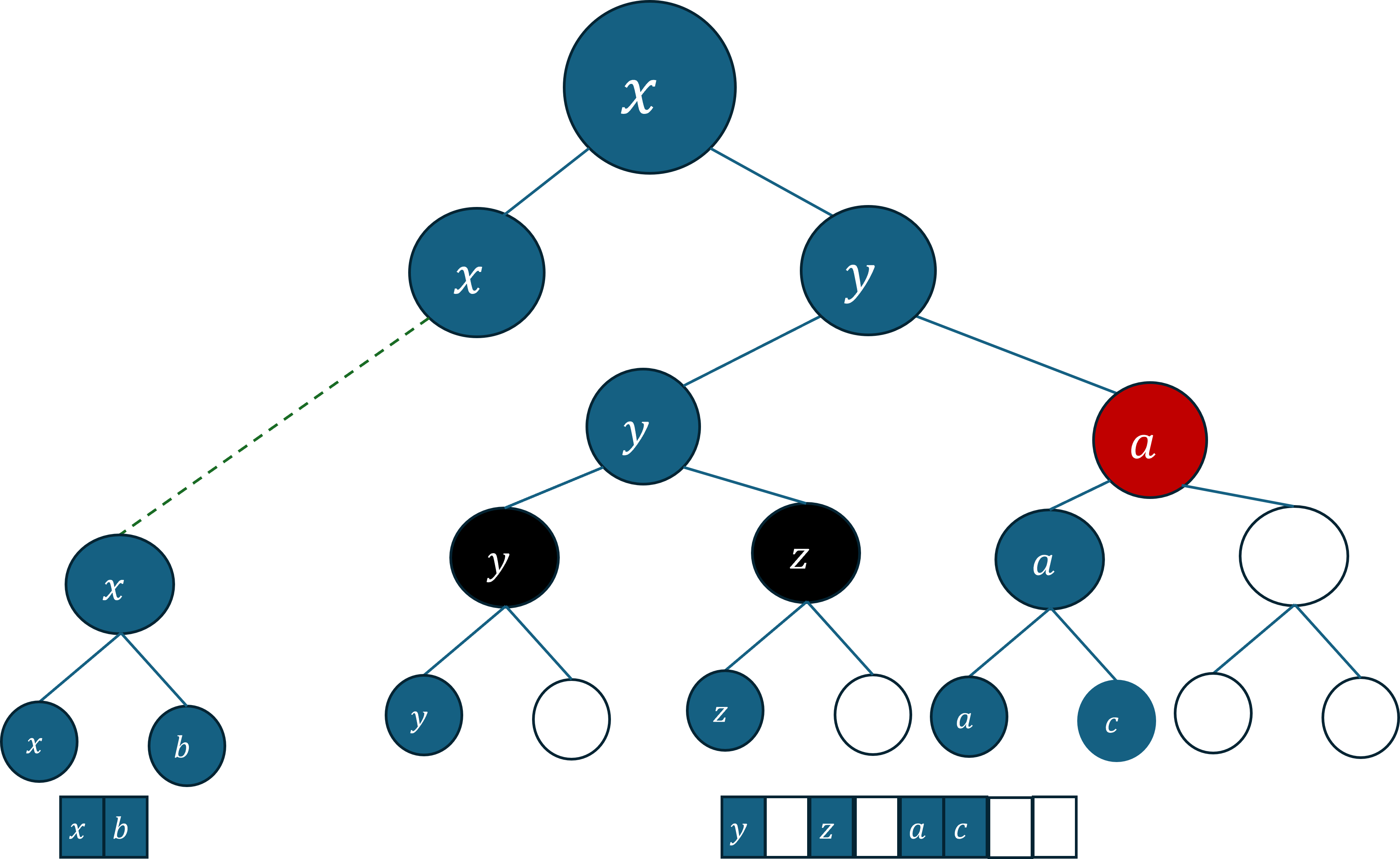} }}
    \hfill
    \subfloat[\centering After the insertion of $x_t$]{{\includegraphics[width=.49\linewidth]{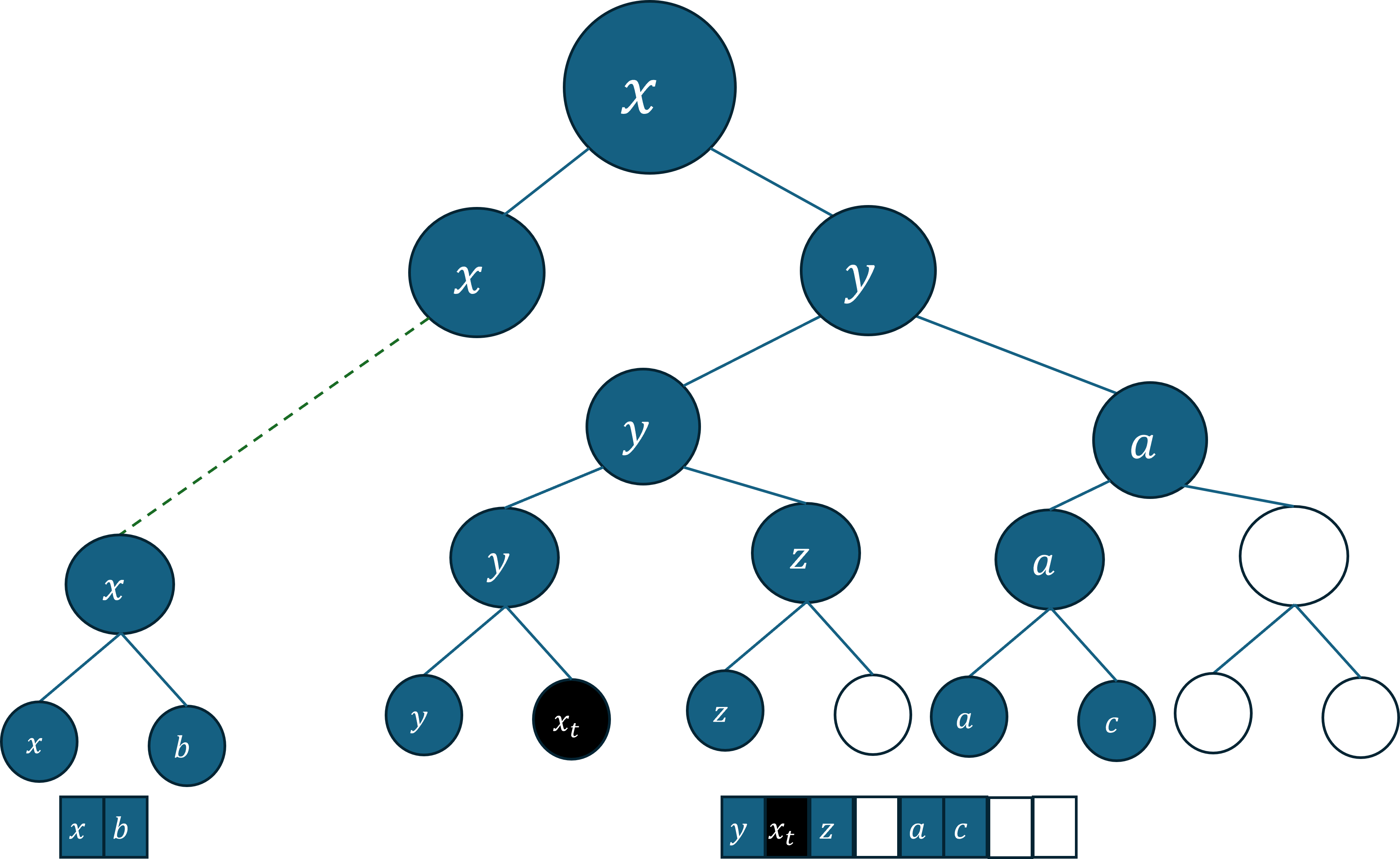} }}
    \caption{Insertion of a point $x_t$ in an \mbc tree. The filled parts of the array represent occupied regions. The colored tree nodes represent marked nodes. On the left figure, suppose the nodes colored black and red are the partial nodes whose balls contain the new point $x_t$. Since the red node is at a higher level than the black nodes, only the black nodes are candidate locations for inserting $x_t$. On the right figure, we show the \mbc tree and the array after inserting $x_t$.}
    \label{fig:insertion to mbc}
\end{figure}

\subsection{Space Usage of the Algorithm}
We start by showing that the size of the array used by the online metric TSP algorithm is at most $(1+\eps)n$. For this, we bound the number of unused cells in the array. Note that each partial node has an unmarked child with an unused subarray under it. Conversely, the unused space in the array can be partitioned into subarrays that are subtended by the unmarked children of partial nodes. Therefore, it suffices to bound the cumulative unused space under unmarked children of partial nodes in the data structure. 

Our main claim is the following: 
\begin{claim}
    \label{clm:allocation stops when the empty space is large}
   For all partial nodes at a fixed height $h$, the cumulative unused space under their unmarked children is at most $2^{H-1}$.
\end{claim}

In order to prove this claim, we will use the following important property.
\begin{claim}\label{cl:small-space-new}
    Let $u$ and $v$ be two partial nodes (nodes with exactly one marked child) at the same height $h$, with centers $c_u$ and $c_v$. Then $d(c_u, c_v) > r_h$.
\end{claim}

\begin{proof}
    Without loss of generality, assume $u$ was marked before $v$ (created by the insertion of a point $x_v$). When the algorithm attempted to insert $x_v$, it checked all existing marked nodes at height $h$, including $u$. Since $u$ is partial, it had an available (unmarked) child, which means the insertion into $u$ would have succeeded if $x_v$ were admissible.
    
    The fact that $x_v$ was not inserted in the subtree under $u$, while the insertion algorithm prefers the partial node at the lowest possible
height, implies that $x_v$ is inadmissible at $u$. Therefore, $d(x_v, c_u) > r_h$. Since $c_v = x_v$, it follows that $d(c_v, c_u) > r_h$.
\end{proof}

Using \Cref{cl:small-space-new}, we now prove \Cref{clm:allocation stops when the empty space is large}.

\begin{proof}[Proof of \Cref{clm:allocation stops when the empty space is large}.]
    Let $N_h$ be the number of partial nodes at height $h$. By \Cref{cl:small-space-new}, their centers are $N_h$ points that are pairwise separated by distance $> r_h$. 
    Since the optimal TSP tour (of cost $\opt$) visits all points, it must connect these $N_h$ centers incurring a cost $> (N_h-1) \cdot r_h$.
    Substituting $r_h = R \cdot 2^{h-H} = \opt \cdot 2^{h-H}$, we get $N_h \leq 2^{H-h}$.
    
    Since each partial node at height $h$ has an unmarked child node at height $h-1$ that has $2^{h-1}$ unused array cells under it, the total unused space in the array is at most
    \[ 
        N_h \cdot 2^{h-1} \leq 2^{H-h} \cdot 2^{h-1} = 2^{H-1}. \qedhere
    \]
\end{proof}

We can now derive the space usage of the online algorithm from \Cref{clm:allocation stops when the empty space is large}.
\begin{lemma}\label{lem:smallspace-unused-space-metric}
    The unused space in the array is at most $\eps n$ at any time. As a consequence, the array uses at most $(1+\eps)n$ space after inserting all $n$ points.
\end{lemma}
\begin{proof}
    By \Cref{clm:allocation stops when the empty space is large}, it follows that the total empty space is at most:
\[
    H\cdot 2^{H-1}
    \le \frac{\eps n}{\log{n}} \cdot  \lfloor\log({\eps n) - \log{\log{n }}}\rfloor
   \le \eps n. \qedhere
\]
\end{proof}

\subsection{Competitive Ratio of the Algorithm}

We want to prove the following bound:
\begin{lemma}\label{lem:cost-smallspace-metric}
    The total cost of the solution produced by the algorithm is $O\left(\frac{\log^2 n}{\eps}\right) \cdot \opt$.
\end{lemma}

To prove this lemma, we first establish some important properties of the online algorithm.
\begin{claim}\label{cl:distance between direct child}
    Let $u,v$ be two marked nodes in the \mbc tree such that $v$ is the parent of $u$. Then, the point $c_u$ at the center of the ball labeling $u$ must be admissible at $v$.
\end{claim}
\begin{proof}
    There are two possibilities. The first is that $c_v = c_u$. In this case, the claim trivially holds. Otherwise, the labeling at $u$ was done by the procedure $\ins(v, c_u)$, which can only be run if $c_u$ is admissible at $v$.
\end{proof}
\begin{claim}\label{cl:label}
    Let $x$ be a point that is stored in the array cell at a leaf node $u$ of an \mbc tree, and let $v$ be a node of height $h$ which is an ancestor of $u$. Then $d(x,c_v) \leq \sum_{k=1}^{h} r_k \leq 2 \cdot r_h$.
\end{claim}

\begin{proof}
    We will prove the claim by induction on $h$.
    For the base case of $h=1$, we have
    from \cref{cl:distance between direct child} that $x$ is admissible at $v$. Therefore, $d(x,c_v) \leq r_1$.

    Let $v_0 = u, v_1, \ldots, v_{h-1}, v_h = v$ denote the vertices on the path from a leaf $u$ to its ancestor $v$ at height $h$. Let $z$ denote $c_{v_{h-1}}$. By the inductive hypothesis, we have $d(x, z) \le \sum_{k=1}^{h-1} r_k$. Furthermore, using \Cref{cl:distance between direct child}, we know that $d(z, c_v) \le r_h$.
    Therefore, by the triangle inequality, 
    \[
        d(x, c_v) \le d(x, z) + d(z, c_v) \le  \sum_{k=1}^h r_k \le 2\cdot r_h. \qedhere
    \]
\end{proof}

Using the above claim, we establish a bound on the total cost within an \mbc tree.

\begin{lemma}\label{lem:general-elementary-cost}
    Suppose a set of points is inserted into an \mbc tree $A$ of height $H$. Then, the total cost of the points in $A$ is at most $4H \cdot \opt$.
\end{lemma}

\begin{proof}
    Let $x_i, x_{i+1}$ be a pair of points that occupy consecutive non-empty cells in the array (the leaves of $A$). Let $v$ denote the least common ancestor of $x_i$ and $x_{i+1}$ in $A$, and let $v_\ell, v_r$ be the left and right children of $v$.
    
    By the structure of the tree, $x_i$ is the rightmost non-empty leaf in the subtree rooted at $v_\ell$, and $x_{i+1}$ is the leftmost non-empty leaf in the subtree rooted at $v_r$. 
    By  \Cref{cl:label}, we have $d(x_i, c_v) \leq 2 r_h$ and $d(c_v, x_{i+1}) \le 2 r_h$. Therefore,     
    by the triangle inequality, we get
    \[ d(x_i, x_{i+1}) \leq d(x_i, c_v) + d(c_v, x_{i+1}) \leq 2 r_h + 2 r_h = 4 r_h. \]
    
    We charge this distance $d(x_i, x_{i+1})$ to the node $v$. Every internal node in the \mbc tree is charged at most once. Summing over all heights $h=1 \dots H$, we get that the total cost is at most
    \[ 
        \sum_{h=1}^{H} 2^{H-h} \cdot 4(\opt \cdot 2^{h-H}) = \sum_{h=1}^{H} 4 \cdot \opt = 4 H \cdot \opt. \qedhere
    \]
\end{proof}

We now return to the proof of \Cref{lem:cost-smallspace-metric}.

\begin{proof}[Proof of \Cref{lem:cost-smallspace-metric}]
    By \Cref{lem:smallspace-unused-space-metric}, the total space used by the algorithm is at most $(1+\eps) n$. The size of the subarray under a single \mbc tree of height $H$ is
    \[
        2^H = 2^{\lfloor\log (\eps n) - \log \log n \rfloor} \geq 2^{\log (\eps n) - \log \log n - 1} = \frac{\eps n}{2 \log n}.
    \]
    Therefore, the number of \mbc trees, denoted $k$, is at most
    \[
        k \le \frac{(1+\eps)n}{\frac{\eps n}{2\log n}} = 2\left(1+ \frac{1}{\eps}\right) \log n
        \le \frac{4\log n}{\eps} .
    \]
    
    We now sum the costs. The total cost consists of the cost within each \mbc tree and the cost incurred between \mbc trees.
    
    \textbf{1. Cost within \mbc trees:} By \Cref{lem:general-elementary-cost}, the cost of a single \mbc tree is at most $4H \cdot \opt$. Since $H \le \log(\eps n) \le \log n$, the cumulative internal cost across all $k$ trees is at most:
    \[
       k \cdot (4H \cdot \opt) \le \frac{4\log n}{\eps} \cdot (4 \log n \cdot \opt) = \frac{16 \log^2 n}{\eps} \cdot \opt.
    \]
    
    \textbf{2. Cost between \mbc trees:} We are left with the cost incurred \emph{between} \mbc trees, i.e., by the rightmost point in $A_i$ and the leftmost one in $A_{i+1}$. For any pair of points, the distance is at most $\opt$. Thus, the total cost between trees is at most
    \[
        k \cdot \opt \le \frac{4\log n}{\eps} \cdot \opt.
    \]
    
    Adding these two types of cost, we get that the total cost is at most
    \[
        \frac{16 \log^2 n}{\eps} \cdot \opt + \frac{4\log n}{\eps} \cdot \opt = O\left(\frac{\log^2 n}{\eps}\right) \cdot \opt. \qedhere
    \]
\end{proof}

%% file: arxiv/small-error-lower-bound.tex
Recall from \Cref{sec:intro} that our lower bound that establishes \Cref{thm:lower} for online metric TSP is actually for the online sorting problem, i.e., on the line metric, Moreover, the lower bound explicitly assumes that the algorithm does not know the value of $\opt$, and therefore, does not apply to prior results in online sorting that assume knowledge of $\opt$.

Before proceeding further, we restate \Cref{thm:lower} for the online sorting problem, and also change notation in a way that will make it easier to describe the construction. The reader can easily verify that the theorem below implies \Cref{thm:lower}.

\begin{theorem}\label{thm:lower-sorting}
    Consider the online sorting problem on $n$ elements. Any deterministic online sorting algorithm that is $c/16$ competitive requires at least $n^{1+1/(2c)}/5$ space, for any $c \le \frac{1}{6}\cdot \frac{\log n}{\log\log n}$.
\end{theorem}

We will assume the terminology of an {\em adversary} and an arbitrary deterministic {\em algorithm}. In every step of the online instance, the adversary reveals a batch of new elements, which the algorithm inserts in empty cells of the array. This sequence repeats itself until the adversary decides to end the instance.

\smallskip\noindent{\bf Epochs and Phases.}
Overall, the online steps are organized in a sequence of {\em epochs}. Conceptually, the adversary achieves two properties in each epoch. First, it increases the range of elements from the previous epoch by a multiplicative factor. Second, because of the change in the overall range, all elements in the previous epochs become much smaller than all elements in the current range. In effect, this allows us to treat the elements in previous epochs as (being close enough to) $0$, which helps simplify the analysis. More precisely, the elements presented by the adversary in the $i$th epoch are in the range $[2^{2i-1}, 2^{2i}]$. Note that by induction, all elements in the previous $i-1$ epochs are in the range $[2^1, 2^{2i-2}]$, which are at least $2^{2i-2}$ far from every element in epoch $i$. 

An epoch is further subdivided into a sequence of {\em phases}. In each phase, the adversary presents a batch of elements, which are inserted by the algorithm in the array. 
We now describe the $j$th phase of the $i$th epoch. As stated previously, the range of elements in this epoch is $[2^{2i-1}, 2^{2i}]$. Let us call this range $R_i$. For $j=0$, the adversary presents $2$ elements, $2^{2i-1}$ and $2^{2i}$. Subsequently, for every $j\ge 1$, the adversary presents a batch of $n^{j/c}-n^{{(j-1)}/c}$ elements {\em that are all distinct} such that the elements in all phases up to the $j$th one are {\em equally spaced in the range $R_i$}.  In particular, the $k$th element of the $i$th epoch at the end of the $j$th phase is given by $2^{2i-1} + k\cdot \frac{2^{2i-1}}{n^{j/c}}$ where $k$ ranges from $0$ to $n^{j/c}$. Note that any two consecutive elements in this epoch differ in their values by $\frac{2^{2i-1}}{n^{j/c}}$. In other words, the $j$th phase supplies the missing elements such that the complete set accumulated by the end of the phase is given by $2^{2i-1} + k\cdot \frac{2^{2i-1}}{n^{j/c}}$, where $k$ ranges from $0$ to $n^{j/c}$ (We give a pictorial depiction of the elements arriving in different phases of an epoch in \Cref{fig:phases}.)

\begin{figure}[tbh]
    \centering
    {{\includegraphics[width=\linewidth]{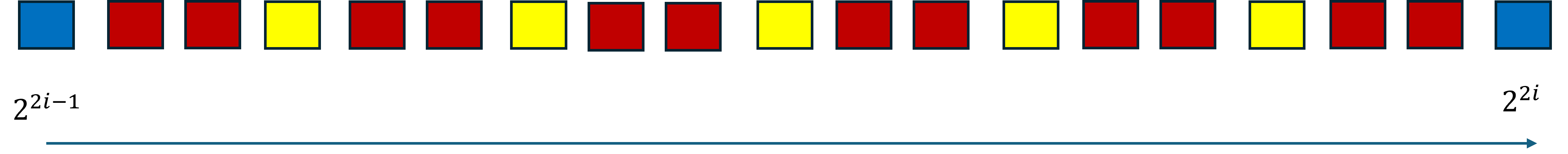} }}
    \caption{Distribution of elements arriving in different phases of epoch $i$. Initially, the two (blue) elements at the ends of the range for the current epoch arrive. The next two phases (yellow and red respectively) denote subsequent arrivals, illustrating the distribution of the elements over time as the epoch progresses.}
    \label{fig:phases}
\end{figure}
{\bf Real and Vacuous Holes.}
This brings us to the crucial decision taken by the adversary after the algorithm inserts the elements presented in a phase. The adversary can perform three possible actions:
\begin{itemize}
    \item[-] it can continue to the next phase of the current epoch,
    \item[-] it can terminate the current epoch and go to the first phase of the next epoch, or
    \item[-] it can terminate the instance altogether.
\end{itemize}
To describe the conditions for this decision, we need to define the notion of holes in the array: 

A {\em hole} is a set of contiguous cells in the array all of which are vacant. 

Initially, the whole array is a single hole.
Over time, as the array fills up, this single hole gives rise to multiple holes. For notational convenience, we extend this terminology to say that two consecutive array cells that are both occupied also have a {\em vacuous} hole between them. In contrast, a hole that comprises an actual non-empty sequence of vacant array cells is called a {\em real} hole. 

It will be useful to define categories of holes, both vacuous and real.
First, we categorize holes into three types based on the elements at their two ends:
\begin{itemize}
    \item[-] {\em Mixed Holes.} These are holes where one end of the hole is occupied by an element from a previous epoch and the other end is occupied by an element from the current epoch.
    \item[-] {\em New Holes.} These are holes where both ends of the hole are occupied by elements from the current epoch.
    \item[-] {\em Old Holes.} These are holes where both ends of the hole are occupied by elements from a previous epoch.
\end{itemize}
Second, orthogonally to the previous categorization, we group holes based on their size:
\begin{itemize}
   \item[-] {\em Large Holes.} These are holes that contain at least $n^{1/(2c)}$ (vacant) array cells. 
   (By definition, all large holes are real holes.)
    \item[-] {\em Small Holes.} These are holes that contain fewer than $n^{1/(2c)}$ (vacant) array cells. 
\end{itemize}
We define the phase of a hole based on its endpoints: a hole is considered to be \emph{created in phase $j$} of the current epoch if at least one of the elements defining its ends belongs to phase $j$.

After the algorithm inserts the elements in the $j$th phase of the $i$th epoch, the adversary makes its next move according to the current state of the array. In particular, it uses the following conditions, which are checked in order:
\begin{itemize}
    \item[-] Condition 1: If the array contains at least $c$ mixed holes (including both real and vacuous holes), then the adversary terminates the instance.
    \item[-] Condition 2: If the value of $j$ reaches $c/2 + 1$, then the adversary terminates the instance.
    \item[-] Condition 3: If the array contains at least $n^{j/c}/2$ large new holes, then the adversary terminates the current epoch and goes to the first phase of the $(i+1)$st epoch.
    \item[-] If none of the above conditions hold, then the adversary continues with the $(j+1)$st phase of the current epoch.
\end{itemize}

To establish the lower bound in \Cref{thm:lower-sorting}, we show two facts: first, that the instance ends with at most $n$ elements, and second, that the competitive ratio is at least $c/16$. 

\smallskip\noindent{\bf Upper Bound on Number of Elements.}
To show that the instance ends with at most $n$ elements, we first lower bound the number of large holes: 
\begin{claim}\label{cl:many-holes}
    Let $n_i$ denote the total number of elements already inserted at the time that the adversary switches from epoch $i$ to $i+1$. Then, the number of large holes is at least $n_i/4$.
\end{claim}
\begin{proof}
We prove this claim by induction on $i$.

Suppose the current epoch $ i$ ended with $j$ phases. Let $n_{curr} = n^{j/c} + 1$ denote the number of elements inserted in the current epoch. 
Since the adversary created a new epoch $i+1$, by Condition 3, the number of large new holes at the end of epoch $i$ is at least $n_{curr}/2$. Now, note that by the inductive hypothesis, the number of large holes at the end of epoch $i-1$ was at least $n_{i-1}/4$. All these holes become old holes at the beginning of epoch $i$. If an element is inserted in such an old hole in epoch $i$, then that creates at least $1$ mixed hole in epoch $i$. (In fact, $2$ mixed holes are created unless the old hole is at one end, in which case only $1$ mixed hole is created.) Since Condition 1 was never satisfied, we can conclude that at most $c$ of these old holes have had insertion of elements in epoch $i$. In other words, at least $n_{i-1}/4  - c$ of the large old holes at the beginning of epoch $i$ have been untouched during epoch $i$, and therefore, they continue to be large holes at the end of epoch $i$. 

We now sum the number of large holes at the end of epoch $i$. The count includes the large new holes created during epoch $i$ and the large old holes from the beginning of epoch $i$ that were untouched during epoch $i$. Thus, the total number of large holes at the end of epoch $i$ is at least
    $n_{curr}/2 + n_{i-1}/4 - c \ge n_i/4$,
since
$n_{curr} = n_i - n_{i-1} \geq n^{1/c} +1 \ge 4c$ for $c \leq \frac{1}{6} \cdot \frac{\log n}{\log{\log{n}}}$. This completes the inductive proof.
\end{proof}

Next, we show that any single epoch cannot contribute more than $O(\sqrt{n})$ elements. This shows that if the instance has $n$ elements overall, then all but $\sqrt{n}$ of them must be from completed epochs.

\begin{claim}\label{cl:max-inserted-in-epoch}
    The maximum number of elements inserted in an epoch is at most $\sqrt{n}+1$. 
\end{claim}
\begin{proof}

    Note that an epoch ends if it reaches the beginning of the $(c/2+1)$st phase.
    Thus, the number of elements in an epoch is at most the number of inserted elements till the end of the $(c/2)$th phase. Since the total number of elements till the $j$th phase is $n^{j/c} + 1$, it follows that the total number of elements in an epoch is at most $n^{(c/2)/c} + 1 = \sqrt{n}+1$.
\end{proof}
Using the above claims, we now show that the instance cannot have more than $n$ elements.
\begin{lemma}\label{lem:terminate}
    The number of elements in the instance constructed by the adversary has at most $n$ elements.
\end{lemma}
\begin{proof}
    Assume for contradiction that $n+1$ elements were inserted. Then, by \Cref{cl:max-inserted-in-epoch}, at least $n-\sqrt{n}$ elements were inserted in completed epochs. Then, by \Cref{cl:many-holes}, there are at least $(n-\sqrt{n})/4$ large holes at the end of the last completed epoch. Since each large hole has $n^{1/(2c)}$ vacant array cells, it follows that the size of the array is at least
    \[
        \frac{n-\sqrt{n}}{4}\cdot n^{1/(2c)}
        > \frac{n^{1+ 1/(2c)}}{5},
    \]
    for $n > 25$. This contradicts the fact that the array is of size $n^{1+1/(2c)}/5$ in the statement of \Cref{thm:lower-sorting}.
\end{proof}

\smallskip\noindent{\bf Lower Bound on the Competitive Ratio.}
Finally, we show that the competitive ratio is at least $c/16$. 
We start with the case when the instance is terminated by the first condition:

\begin{lemma}\label{lem:old-holes-cost}
    If the array contains at least $c$ mixed holes, then the competitive ratio is at least $c/4$.
\end{lemma}
\begin{proof}
Recall that $R_i = [2^{2i-1}, 2^{2i}]$ denotes the range for the current epoch $i$. The elements in all the previous $i-1$ epochs are in the range  $\cup_{j=1}^{i-1} [2^{2j-1}, 2^{2j}] \subseteq [2^1, 2^{2i-2}]$. As a consequence, the difference between any element in the current epoch and any element in a previous epoch is at least $2^{2i-1} - 2^{2i-2} = 2^{2i-2}$. Therefore, every mixed hole induces a cost of at least $2^{2i-2}$.

Now, note that optimal cost is at most the maximum value in the current range $R_i$, i.e., $\opt \le 2^{2i}$. Since there are at least $c$ mixed holes, the total cost is at least $c\cdot 2^{2i-2}$, which implies a competitive ratio of at least $c/4$.
\end{proof}

Next, we consider the case where the instance is terminated by the second condition. Note that in this case, the third condition did not hold at the end of the $(c/2)$th phase that immediately preceded the current phase. (Otherwise, the adversary would have moved to the next epoch.) We establish the following claim, which yields a lower bound of the cost of the algorithm based on the index of the current phase.

\begin{claim}\label{cl:few-holes}
    Suppose $c \leq \frac{1}{6} \cdot \frac{\log{n}}{\log{\log{n}}}$. Then, if the array contains fewer than $n^{j/c}/2$ large new holes after inserting the elements in the $j$th phase of epoch $i$, then the cost of the algorithm is at least $j\cdot 2^{2i}/8$.
\end{claim}

\begin{proof}
In this proof, we obtain a lower bound on the cost by only considering small new holes, and disregarding all other types of holes. Specifically, in each phase, we account only for the cost of small holes newly created during that phase. Once a small new hole is created in a phase, further subdivisions of this hole in subsequent phases are not counted, since these might not cause additional cost.

Suppose we are at the end of the $j$th phase. If epoch $i$ did not end after phase $j$, then conditions 1 and 3 did not hold at the end of this phase. Therefore, there can be at most $n^{j/c}/2$ large new holes, in addition to at most $c$ mixed holes. Since $n^{j/c} + 1$ elements were inserted in the current epoch, the number of small new holes is at least:
\begin{equation}\label{eq:small}
    n^{j/c} -  n^{j/c}/2 - c = n^{j/c}/2 - c \geq n^{j/c}/3,
\end{equation}
where the inequality holds for any $j\ge 1$, since $n^{1/c} \ge 6c$ for $c \leq \frac{1}{6} \cdot \frac{\log{n}}{\log{\log{n}}}$.

Observe that inserting an element into an empty cell within an existing small new hole does not incur any additional cost if the element's value lies between those of the elements at its two ends. Consequently, to derive a valid lower bound for the cost in phase $j$, we exclude from our analysis all elements inserted into small holes that were created during the first $j-1$ phases. However, if an element is inserted in an existing large new hole, and this creates a small new hole, this does increase the cost since we explicitly excluded the cost of all large holes from the analysis.
Likewise, if the insertion of elements into a mixed hole creates a small new hole, we also include its cost in our accounting.

So, we want to discard small new holes of two types from the above count: (a) those that have elements from a previous phase before $j$ at either end, and (b) those that were created in the current phase $j$ by inserting elements in phase $j$ into a small new hole that existed at the end of phase $j-1$. We count these together by counting the total number of cells that were one of two types at the end of phase $j-1$: (a) either a vacant cell in a small new hole or (b) an occupied cell at one end of a small new hole.

The number of elements inserted in the current epoch up to phase $j-1$ is given by:
\[
2 + \sum_{k=1}^{j-1} \left( n^{k/c} - n^{(k-1)/c} \right) = 2 + n^{(j-1)/c} - n^0 = n^{(j-1)/c} + 1.
\]
Therefore, the number of small new holes at the end of phase $j-1$ is at most $n^{(j-1)/c}$,
since each new hole has an element from the current epoch on each side.
Note that each small new hole has at most $n^{1/(2c)}$ vacant cells. 
Therefore, the total number of cells that are either vacant in a small new hole or 
at one end of a small new hole at the end of phase $j-1$ is at most:
\begin{equation}\label{eq:double-count}
n^{(j-1)/c} \cdot (n^{1/(2c)} - 1 + 2) = n^{(j-1/2)/c} + n^{(j-1)/c} 
\le 2 n^{(j-1/2)/c}.
\end{equation}

Subtracting \eqref{eq:double-count} from \eqref{eq:small}, we get that the 
number of small new holes that incur additional cost in the $j$th phase is at least:
\[
    n^{j/c}/3 - 2 n^{(j-1/2)/c} \geq n^{j/c}/4,
\]
where the inequality holds since $c \leq \frac{1}{6} \cdot \frac{\log n}{\log{\log{n}}}$.

Next, we determine the minimum cost per hole. The elements in range $R_i$ are uniformly spaced such that the minimum difference between any two distinct elements in phase $j$ is
\[
\frac{2^{2i}-2^{2i-1}}{n^{j/c}} = \frac{2^{2i-1}}{n^{j/c}}.
\]
Multiplying the number of contributing holes by the minimum cost per hole, the total cost accumulated strictly within phase $j$ is at least
\[
\frac{2^{2i-1}}{n^{j/c}} \cdot \frac{n^{j/c}}{4} = \frac{2^{2i-1}}{4} = 2^{2i-3}.
\]
Multiplying by the number of phases, the total cost is at least $j \cdot 2^{2i-3}$.
\end{proof}

Using the above claim, we immediately get the following:

\begin{lemma}\label{lem:large-phase}
    If the algorithm terminates with $j = c/2 + 1$, then the competitive ratio of the algorithm is at least $c/16$. 
\end{lemma}
\begin{proof}
    This is an immediate corollary of \cref{cl:few-holes}. Specifically, if $j$ reaches $c/2+1$, then the cost at the end of the last completed phase is at least $\frac{c/2}{8} \cdot \opt = \frac{c}{16} \cdot \opt$. 
\end{proof}

%% file: arxiv/closing-remarks.tex
In this paper, we explored the tradeoff between competitive ratio and space usage for the online metric TSP problem. In particular, we gave a deterministic algorithm that uses $(1+\eps)n$ space and improves the competitive ratio from $\Theta(\sqrt{n})$ (for $n$ space) to $O(\log^3 n/\eps)$. We also showed that this cannot be improved further to $O(1)$-competitiveness using a deterministic algorithm, unless the space used increases to $n^{1+\gamma}$ for constant $\gamma > 0$.

Our work raises several interesting questions. First, our lower bound only applies to deterministic algorithms; we believe that the lower bound might also hold for randomized algorithms, but the current construction does not readily extend to the randomized setting. Second, while we rule out $O(1)$-competitive deterministic algorithms with $m= O(n\cdot \polylog(n))$, it is quite possible that setting $m = n^{1+\eps}$ for some constant $\eps > 0$ yields an $O(1)$-competitive algorithm. Indeed, improving on the competitive ratio of $O(\log n)$~\cite{RosenkrantzSL77,ImaseW91}, even with unlimited space usage, remains open.

%% file: arxiv/general-metric-doubling.tex
In this section, we extend the result in \Cref{sec:elementary} to the case of unknown $\opt$, at the cost of an additional factor of $O(\log n)$ in the competitive ratio. 

Following \cite{AzarPV26}, we apply a doubling scheme on the value of the optimal solution. The doubling scheme uses two parameters: the optimal cost $\opt$ and the number of inserted points for a given guess on $\opt$. We have an outer loop that (at least) doubles the estimate on $\opt$ in each iteration, and an inner loop that doubles the number of inserted points. 
This nested doubling process is initialized after sequentially inserting the first two points (in the first two array cells); twice their distance provides an initial estimate for $\opt$.

Each iteration of the outer loop is called an epoch. 
Since \mbc trees are not defined for small values of $n$ (see \Cref{rem:epsilon-cannot-be-too-small}), we start any epoch by inserting the first $\frac{\log^2{1/\eps}}{\eps}$ points consecutively before starting the doubling scheme.
In the $i$-th epoch, the guess on $\opt$ is denoted $\opt_i$. When the arrival of a new point causes the estimated cost (defined as twice the cost of the minimum spanning tree on the points seen so far) to exceed $\opt_i$, the $i$-th epoch ends and the $(i+1)$-st epoch starts. The value of $\opt$ in the new epoch is set to be at least double that of the previous epoch. Each epoch is assigned new space in the array, which begins immediately after the last cell used in the previous epoch.

Within an epoch, each iteration of the inner loop is called a phase.
Let $n_{i,j}$ denote the bound on the number of inserted points in the $j$-th phase of the $i$-th epoch. The phase ends when the number of points exceeds $n_{i,j}$, and the $(j+1)$-st phase starts with $n_{j+1} := 2 n_{i,j}$. Overall, we define $n_{i,j} := 2^j$.

In the $j$-th phase, we allocate a new subarray that begins immediately after the subarray used in the $(j-1)$-st phase. The amount of space allocated in the $j$-th phase is $(1+\eps/3)n_{i,j}$.
During the $j$-th phase, we employ the algorithm from \Cref{sec:elementary} for inserting $n_{i,j}$ points with the variable $\eps / 3$ (instead of $\eps$) as a black box.

In the proofs below, we use \(k_i, n_i\) to respectively denote the number of phases and points in the $i$-th epoch, and \(n_{i,j}\) to denote the number of points in its \(j\)-th phase.
\begin{lemma}[Space Bound]\label{lem:unknown-small-space-space}
    The algorithm inserts $n$ points in $(1+\eps)n$ space.
\end{lemma}
\begin{proof}
By \Cref{lem:smallspace-unused-space-metric}, in phases $\le k_i-1$, the unused space is at most $\eps/3$ times the used space, and in phase $k_i$, it is at most $\eps/3$ times $n_{i, k_i}$. But, $n_{i, k_i} \le 2 n_{i, k_{i-1}}$, i.e., the unused space in phase $k_i$ is at most $2 \eps/3$ times the used space in this epoch. Thus, the total unused space is at most an $\eps$-fraction of the used space in this epoch. The lemma follows. 
\end{proof}
\begin{lemma}[Cost Bound]
    The cost incurred by the algorithm is at most $\opt\cdot O(\log ^3 n / \eps)$.
\end{lemma}
\begin{proof}
Fix epoch $i$. The cost between consecutive phases is at most \(\opt_i\), based on the current bound on the optimal solution. This results in a total additive cost of at most \(k_i \cdot \opt_i\) between the $k_i$ phases. In addition, since we insert the first $\frac{\log^2 (1/\eps)}{\eps}$ points sequentially, we incur an additional cost of at most $\frac{\log ^2 (1/\eps)}{\eps} \cdot \opt_i$. Finally, there is a cost of $O\left(\frac{\log^2 n_{i, j}}{\eps}\right)\cdot \opt_i = O(j^2/\eps)\cdot \opt_i$ within each phase by \Cref{lem:cost-smallspace-metric}.
Therefore, the total cost in the \(i\)-th epoch is at most
\[
\opt_i \cdot O\left(\frac{\log^2 (1/\eps)}{\eps} + k_i + \sum_{j=1}^{k_i} \frac{j^2}{\eps} \right)
= \opt_i \cdot O\left(\frac{\log^2 (1/\eps)}{\eps} + \frac{k_i^3}{\eps} \right). 
\]
In addition, the cost between epochs $i-1$ and $i$ is at most $\opt_i$.
Summing over all $\ell$ epochs, we get that the total cost is at most
\[
\sum_{i=1}^\ell \opt_i \cdot O\left(\frac{\log^2 (1/\eps)}{\eps} + \frac{k_i^3}{\eps} \right).
\]

Note that $\opt_i$ at least doubles in every epoch, and at any time it provides a $2$-approximation to the true value of $\opt$.
As a result, $\sum_{i=1}^\ell \opt_i \le 2 \opt_\ell \le 4\opt$. Moreover, the total number of phases in any epoch is at most $\log n$, i.e., $k_i \le \log n$ for every $i$. Applying these facts to the above cost bound, we get that the total cost is at most
\[
   \opt \cdot O\left(\frac{\log^2 (1/\eps)}{\eps} + \frac{\log^3 n}{\eps} \right)
    = \opt \cdot O(\log^3 n / \eps). \qedhere
\]
\end{proof}

%% file: arxiv/counterexample.tex
The upper bound analysis in \cite{AzarPV26} relies on the structural property that all partial nodes at any height $h$ are labeled by disjoint intervals. The adaptation of this key property into balls for a general metric is used in our paper as well in \Cref{cl:small-space-new}. However, the insertion algorithm from \cite{AzarPV26}, adapted to general metrics, does not maintain this key property.

First, we recall the insertion algorithm from \cite{AzarPV26}. The procedure $\ins(v, x)$ is defined as follows:
 \begin{itemize}
    \item[-] If $v$ is a leaf node and is unmarked, then label $v$ with the left or right half of its parent's label that contains $x$, write $x$ in the array cell at the leaf node $v$, and return success.
    \item[-] If $v$ is marked and $x$ is inadmissible at $v$, then return failure.
    \item[-] If $v$ is marked and $x$ is admissible at $v$, then call $\ins(v_\ell,x)$. If it returns success, return success. Otherwise, if it returns failure, then call $\ins(v_r,x)$. If this returns success, then return success. Otherwise, if it  returns failure as well, then return failure.
    \item[-] If $v$ is unmarked, then label $v$ with the left or right half of its parent's label that contains $x$ and call $\ins(v_\ell,x)$.
\end{itemize}
In short, $\ins(v, x)$ attempts to insert $x$ first to the left child of $v$ and then to the right child of $v$.
Note that this operation also applies for balls as labels instead of intervals.
As usual, the overall data structure comprises a series of identical \mbc trees $A_0, A_1, \ldots$.
Let $x_t$ be the point being inserted. The algorithm in \cite{AzarPV26} iterates over the subarrays corresponding to these trees in order, i.e., it attempts to insert $x_t$ in $A_0$, then $A_1$, then $A_2$, etc., until the point is successfully inserted.

We now demonstrate that the distance property claimed in \cref{cl:small-space-new} fails for this insertion algorithm for $\mathbb{R}^2$ with the $\ell_1$ metric.

The instance starts by inserting $(1/2,0), (0,1/2)$.
As shown in \Cref{fig:counterexample}, these two points mark two sibling nodes (call these $u, v$) in the \mbc tree with balls of radius $1/2$ and centered at $(1/2,0),(0,1/2)$ respectively. 
Next, we insert the points $(\delta,0), (0,\delta)$ for a small $\delta > 0$.
Observe that $\|(\delta,0) - (1/2,0)\|_1 = |1/2 - \delta| < 1/2$, whereas $\|(\delta,0) - (0,1/2)\|_1 > 1/2$. Thus, $(\delta, 0)$ is inserted under $u$ and creates a partial node (call it $u'$) labeled by a ball centered at $(\delta,0)$. Similarly, $(0,\delta)$ is inserted under $v$ and creates a partial node (call it $v'$) labeled by a ball centered at $(0,\delta)$. The two partial nodes $u', v'$ at the same height contradict \Cref{cl:small-space-new} since the points $(\delta,0), (0,\delta)$ are only at distance $2\delta$ from each other.

\begin{figure}[tbh]
    \centering
    \subfloat[\centering After the insertion of $(1/2,0)$ and $(0,1/2)$]{{\includegraphics[width=.49\linewidth]{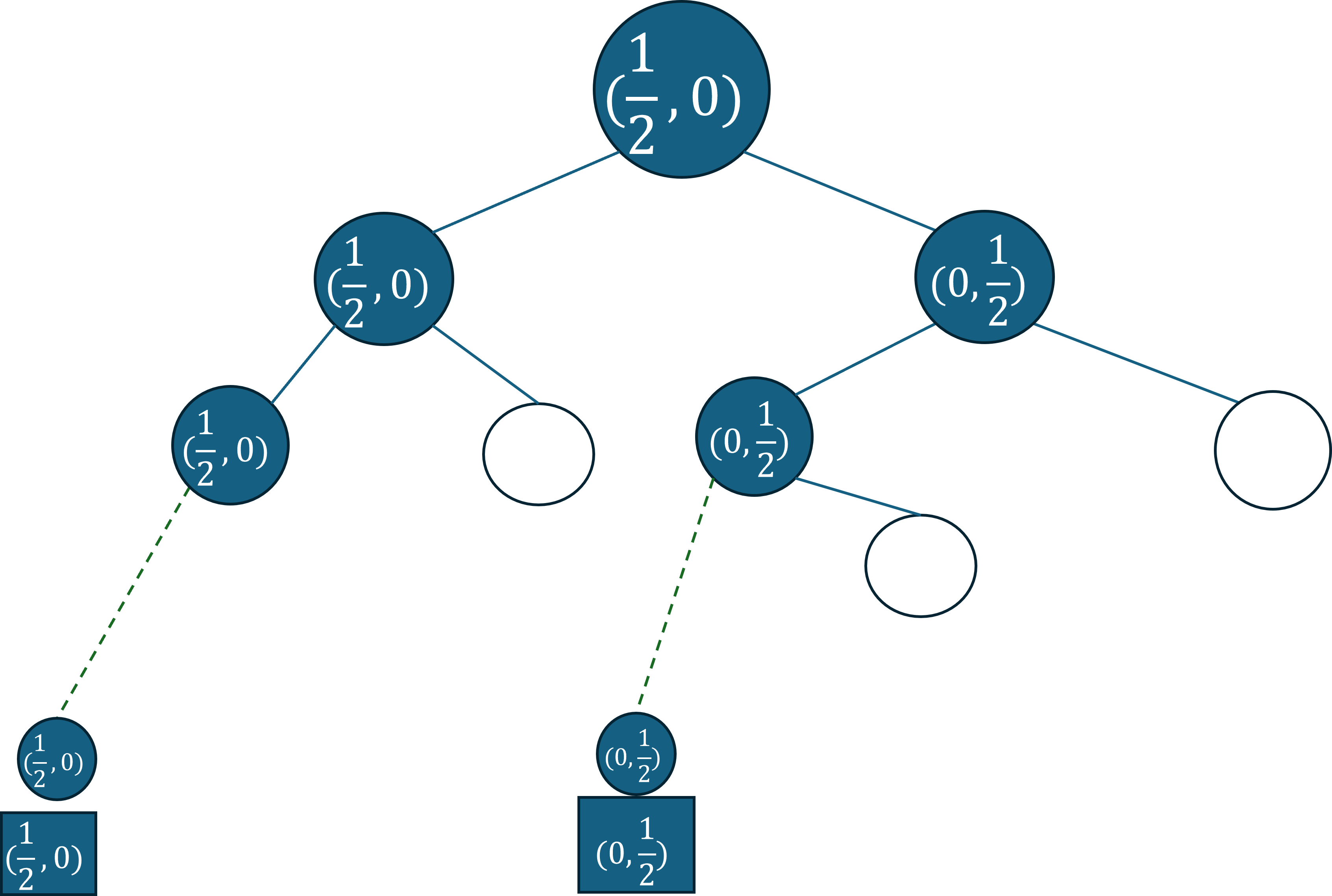} }}%
    \hfill
    \subfloat[\centering After the insertion of $(\delta,0)$ and $(0,\delta)$]{{\includegraphics[width=.49\linewidth]{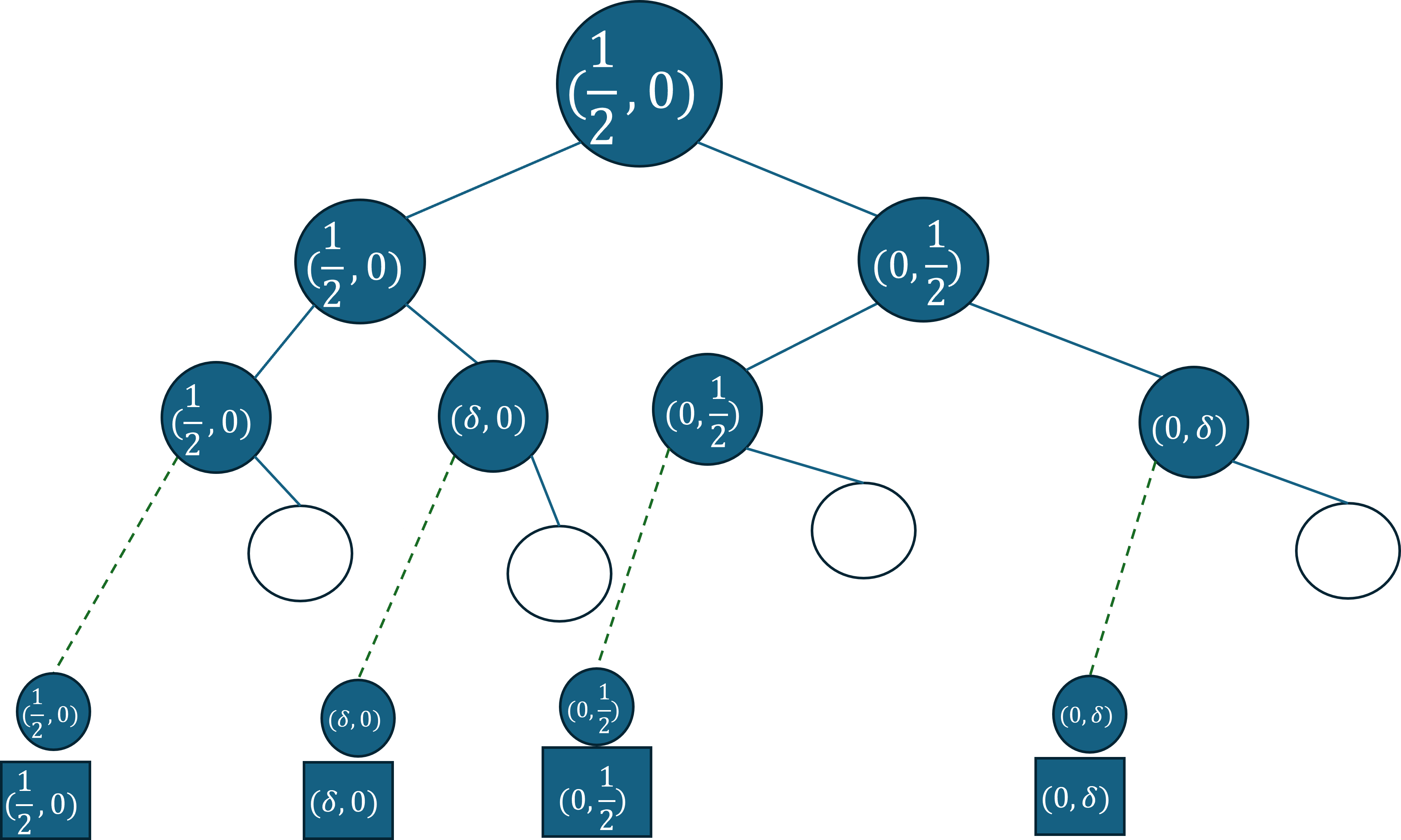} }}%
    \caption{Example showing that the algorithm in \cite{AzarPV26} does not satisfy \Cref{cl:small-space-new} for the $\ell_1$-metric on $\mathbb{R}^2$. On the left, after the initial insertion of $(1/2, 0)$ and $(0, 1/2)$, we have sibling nodes $u, v$ with disjoint balls of radius $1/2$. On the right, after the subsequent insertion of $(\delta, 0)$ and $(0, \delta)$, we create partial nodes $u', v'$ at the same height, but their centers are arbitrarily close, within $2\delta$ of each other. This violates \Cref{cl:small-space-new}.}
    \label{fig:counterexample}
\end{figure}